\documentclass[a4paper,USenglish,cleveref,autoref,thm-restate]{lipics-v2021}
\pdfoutput=1 %ensure pdflatex processing (mandatatory e.g. to submit to arXiv)
\hideLIPIcs  %remove references to LIPIcs series (logo, DOI, ...), e.g. when preparing a pre-final version to be uploaded to arXiv or another public repository
\usepackage{tiz-paper}

\definecolor{AyaFn}{HTML}{00627a}
\definecolor{AyaConstructor}{HTML}{067d17}
\definecolor{AyaStruct}{HTML}{00627a}
\definecolor{AyaGeneralized}{HTML}{00627a}
\definecolor{AyaData}{HTML}{00627a}
\definecolor{AyaPrimitive}{HTML}{00627a}
\definecolor{AyaKeyword}{HTML}{0033b3}
\definecolor{AyaComment}{HTML}{8c8c8c}
\definecolor{AyaField}{HTML}{871094}

\NewDocumentCommand{\stepsTo}{m}{\xmapsto{#1}}

\NewDocumentCommand{\recvEx}{}{ \operatorname{\mathsf{recv}} }
\NewDocumentCommand{\sendEx}{}{ \operatorname{\mathsf{send}} }

\NewDocumentCommand{\procEx}{mm}{ \textsf{proc}[#1](#2) }
\NewDocumentCommand{\procFwd}{mm}{ \textsf{fwd}[#1](#2) }
\NewDocumentCommand{\AtomicObjPrefix}{}{ \textsf{Atomic} }
\NewDocumentCommand{\AtomicObj}{m}{ \AtomicObjPrefix(#1) }

\NewDocumentCommand{\lwith}{}{\mathrel{\&}}
\NewDocumentCommand{\fwdEx}{}{ \operatorname{\mathsf{fwd}} }
\NewDocumentCommand{\letEx}{}{ \operatorname{\mathsf{let}} }
\NewDocumentCommand{\partialmap}{}{\rightharpoonup}

\NewDocumentCommand{\applyCompl}{mm}{ \operatorname{\textsf{applyCompl}}(#1, #2) }

\NewDocumentCommand{\valueInterp}{m}{\mathcal{V}\lrbracket{#1}}
\NewDocumentCommand{\termInterp}{m}{\mathcal{E}\lrbracket{#1}}

\NewDocumentCommand{\Cfg}{}{\mathsf{Cfg}}
\NewDocumentCommand{\NCfg}{}{\mathsf{NCfg}}
\NewDocumentCommand{\Obj}{}{\mathsf{Obj}}
\NewDocumentCommand{\NObj}{}{\mathsf{NObj}}
\NewDocumentCommand{\Labels}{}{\mathbb{A}}
\NewDocumentCommand{\FMSet}{m}{\mathsf{FMSet}(#1)}

\NewDocumentCommand{\Actions}{}{\mathsf{Actions}}
\NewDocumentCommand{\fmconcat}{}{\mathrel{⊎}}

\NewDocumentCommand{\evidence}{m}{\tag*{#1}}

\usepackage[auto]{contour}
\NewDocumentCommand{\vProc}{}{P}
\NewDocumentCommand{\nvProc}{}{\contour{black}{\color{white}$P$}}
\NewDocumentCommand{\vCfg}{}{\Omega}
\NewDocumentCommand{\nvCfg}{}{\contour{black}{\color{white}$\Omega$}}

\NewDocumentCommand{\textcode}{m}{\textsf{#1}}
\NewDocumentCommand{\textlibraryname}{m}{\texttt{#1}}

\NewDocumentCommand{\ProcLang}{}{\textcode{ProcLang}}

\RequirePackage{tikz}
\usetikzlibrary{automata, positioning, arrows}

\NewDocumentCommand{\TODO}{m m}%
  {{\bfseries\color{#1}[#2]}}%

\crefname{example}{Example}{Examples}
\Crefname{example}{Example}{Examples}

\title{A Language-Agnostic Logical Relation for Message-Passing Protocols}
\NewDocumentCommand{\cmuAddress}{}{Carnegie Mellon University, Pittsburgh, USA}

\author{Tesla Zhang}{\cmuAddress}{teslaz@cmu.edu}{https://orcid.org/0000-0002-9050-846X}{}
\author{Sonya Simkin}{\cmuAddress}{ssimkin@andrew.cmu.edu}{https://orcid.org/0009-0008-9261-6318}{}
\author{Rui Li}{\cmuAddress}{ruil3@andrew.cmu.edu}{https://orcid.org/0009-0006-3555-9770}{}
\author{Yue Yao}{\cmuAddress}{yueyao@cs.cmu.edu}{https://orcid.org/0000-0001-8523-5156}{}
\author{Stephanie Balzer}{\cmuAddress}{balzers@cs.cmu.edu}{https://orcid.org/0000-0002-8347-3529}{}

\authorrunning{J. Open Access and J.\,R. Public} %TODO mandatory. First: Use abbreviated first/middle names. Second (only in severe cases): Use first author plus 'et al.'
\Copyright{Jane Open Access and Joan R. Public} %TODO mandatory, please use full first names. LIPIcs license is "CC-BY";  http://creativecommons.org/licenses/by/3.0/


 %TODO mandatory; please add comma-separated list of keywords

\funding{
This material is based upon work supported by the NSF
under Grant No. (2211996 and 2442461)
and upon work supported by the AFOSR
under award number FA9550-21-1-0385 (Tristan Nguyen, program manager).
Any opinions, findings, and conclusions or recommendations expressed in this material are those of the author(s) and do not necessarily reflect the views of
the National Science Foundation or
the U.S. Department of Defense.
}

\acknowledgements{The authors would like to thank Robbert Krebbers and Jules Jacobs on their valuable feedback and discussions on mechanization.}

\nolinenumbers

\EventEditors{John Q. Open and Joan R. Access}
\EventNoEds{2}
\EventLongTitle{42nd Conference on Very Important Topics (CVIT 2016)}
\EventShortTitle{CVIT 2016}
\EventAcronym{CVIT}
\EventYear{2016}
\EventDate{December 24--27, 2016}
\EventLocation{Little Whinging, United Kingdom}
\EventLogo{}
\SeriesVolume{42}
\ArticleNo{23}
\begin{document}

\maketitle

\begin{abstract}
Today's computing landscape has been gradually shifting to
applications targeting distributed and \emph{heterogeneous} systems,
such as cloud computing and Internet of Things (IoT) applications.
These applications are predominantly \emph{concurrent},
employ \emph{message-passing},
and interface with \emph{foreign objects},
ranging from externally implemented code to
actual physical devices such as sensors.
Verifying that the resulting systems adhere to
the intended protocol of interaction is challenging---the
usual assumption of a common implementation language,
let alone a type system,
no longer applies,
ruling out any verification method based on them.
This paper develops a framework for certifying
\emph{protocol compliance} of heterogeneous message-passing systems.
It contributes the first mechanization of a
\emph{language-agnostic logical relation},
asserting that its inhabitants comply with the protocol specified.
This definition relies entirely on a
labelled transition-based semantics,
accommodating arbitrary inhabitants,
typed and untyped alike, including foreign objects.
As a case study, the paper considers two scenarios:
(1) \emph{per-instance verification} of a specific application or
hardware device, and
(2) \emph{once-and-for-all verification} of well-typed applications
for a given type system.
The logical relation and both scenarios are mechanized in
the Coq theorem prover.
\end{abstract}

\section{Introduction}%
\label{sec:intro}
% Explain heterogeneity
Modern computing applications are increasingly becoming \emph{heterogeneous}, with many systems needing to interface with \emph{foreign objects}.
These foreign objects can range from code that is externally implemented (in a language that may or may not be the same as the one internal to the system)
to actual hardware devices which communicate with the system. 
% Example 1: Multi-language
An example of such an application can be found in cloud computing. In this setting, users may outsource resource-intensive computations to a cloud service, which scales the computation across several machines.
% Example 2: Hardware (from Yue's paper)
As another example, consider a smart home system for monitoring air quality. The controller of the system communicates with different
sensors around the home to gather data on surrounding air temperature, humidity, pressure, etc.
Such sensors are \emph{hardware} devices, and the controller must interact with through a protocol defined
by the manufacturer in a specification.

% Problem
Heterogeneity poses as a major challenge when one seeks to certify these systems against a given protocol specification. In particular,
there are three primary obstacles that arise in this setting:
\begin{enumerate}
    \item\label{obstacle1} \emph{Lack of a common specification language}. In a homogeneous setting, we would easily be able to detail a specification in whatever language each of our components is written in.
    However, in the presence of foreign objects, there may be many languages at play with no unified specification language.
    In fact, there may not even be a notion of a ``source language'' for certain objects, as evident with the sensor example.
    \item\label{obstacle2} \emph{Lack of a common computational model}. In a similar vein, we do not have a common method to specify the runtime behavior of 
    arbitrary objects in the system, which can include both programs and physical devices.
    \item\label{obstacle3} \emph{Lack of compositionality}. Compositionality (i.e. modularity) allows for different components of a system to be
    verified independently and combined in a way that guarantees a verified whole, \emph{without} the burden of re-verifying the entire system.
    In addition to facilitating reasoning about a system, compositionality is also crucial to ensuring scalability.
    For an arbitrary heterogeneous system, we may not necessarily have a way to compose the verification of any smaller part.
\end{enumerate}
% Solutions
To combat all three of these issues, we employ three techniques: \emph{types} for \emph{message-passing protocols}, \emph{labelled transition systems}, and \emph{logical relations}.

% Types Background
To address \cref{obstacle1}, we use \emph{types} as a specification language. We are particularly
interested in the verification of \emph{concurrent}, \emph{message-passing} based systems, as those are the ones dominating the current computing landscape.
As such, we choose \emph{behavioral} types \cite{AnconaARITCLE2016,GayRavaraBOOK2017} that are able to express the communications between objects.
This provides us with a unified way to describe specifications of any of the components of a system.

% LTS Background
To address \cref{obstacle2}, we use a \emph{labelled transition system} (LTS) \cite{MilnerBook1980, MilnerBook1999, SangiorgiWalkerBook2001}
to express the message-passing behavior of any program or device in a system. 
In an LTS, each transition on an object is annotated with an \emph{action} that describes the willingness of that object
to engage in communication.
Using an LTS allows us to abstract from the syntactic representation of any particular object in a system, since we are able to describe the
behavior of any component simply by the action that it performs.

% LR background
To address \cref{obstacle3}, we use \emph{logical relations} \cite{GirardPhD1972, PittsStarkHOOTS1998, PlotkinTR1973, StatmanARTICLE1985, TaitARTICLE1967}
as our verification framework. Logical relations are a technique that allows one to prescribe properties of valid programs in terms of their 
\emph{computational behavior}, as opposed to solely their static properties.
Crucial to our development is the concept of \emph{semantic typing} \cite{ConstableBook1986, LoefARTICLE1982, TimanyJACM2024},
which allows terms that are not necessarily (syntactically) well-typed to be an inhabitant of the logical relation.
Such a semantic approach allows one to prove the inhabitance of not only well-typed terms (via the ``fundamental theorem''
of the logical relation), but also of untyped terms which behave ``correctly'' with respect to the definition of the logical relation.
This property has paved the way for multi-language logical relations, as seen in those developed for compiler correctness proofs \cite{BentonICFP2009, ChlipalaPLDI2007, MinamidePOPL1996}
and soundness of language interoperability \cite{PattersonPLDI2022}.
Another crucial aspect of logical relations (particularly for verification) is the fact that logical relations inherently facilitate \emph{compositionality}---any two inhabitants of the logical relation can be composed 
to a compound inhabitant, as dictated by the type structure of the underlying language.
Semantic logical relations are thus particularly well suited to accommodate the heterogeneity of the today's applications.

% Mechanization
In addition to providing the definition of a verification framework for heterogeneous systems, we are equally invested in providing a \emph{mechanization} of such a framework,
which can be found in the following GitHub repository: \url{https://github.com/balzers/LAgnoLR}.
Our community has a long-standing tradition of making mathematically precise definitions of programming languages to admit rigorous reasoning upon them.
However, with the increasing complexity of these definitions and subsequent theorems, our community has come to realize that the use of proof assistants in this process is necessary.
Mechanization not only provides a machine certification of our proofs,
but also scales easily with future extensions,
unlike pen-and-paper proofs.

% Summary of contributions
In this paper, we present the first mechanization of a \emph{language-agnostic} logical relation for certifying protocol compliance of message-passing heterogeneous systems. 
Being language-agnostic means that we are able to verify systems that consist of any kinds of foreign objects, accommodating for the many heterogeneous systems of today.
This work is based on the development in Yue \textit{et al.} \cite{YaoPOPL2025}, which defines a logical relation for time-dependent heterogeneous systems using message-passing types and a labelled transition system.
We generalize this work and provide a mechanization of the logical relation in Coq, as well as the mechanization of two verification scenarios in our framework.

\subparagraph*{Contributions.}

Our contributions are
\begin{itemize}

\item A formalization of a \emph{language-agnostic} logical relation for protocol compliance of message-passing systems. Lifting the burden of syntactic convergence between the components of a system, the relation accommodates arbitrary objects as inhabitants---typed and untyped alike, including actual physical devices.

\item A case study encompassing inhabitance proofs of two kinds: (1) \emph{per-instance verification} of a specific application or
hardware device, and
(2) \emph{once-and-for-all verification} of well-typed applications
for a given type system.

\item A mechanization of the above contributions in the Coq theorem prover.

\end{itemize}

\subparagraph*{Paper structure.}
In \cref{sec:verification-framework}, we introduce and provide intuition for all of the relevant concepts and definitions for our verification framework.
In \cref{sec:case-study}, we showcase the two modes of verification our framework accommodates through the ``per-instance'' verification of a particular object and the development of a ``once-and-for-all'' verification method for a given type system.
In \cref{sec:mechanization}, we discuss the mechanization of the developments in \cref{sec:verification-framework} and \cref{sec:case-study}, and we follow with an exploration of future work (\cref{sec:discussion-future}), contrasting with related work (\cref{sec:related}), and a conclusion (\cref{sec:conlusion}).

\section{Verification Framework}%
\label{sec:verification-framework}
In this section, we provide definitions for the three critical components of our verification framework: 
the \emph{computational model}, which is defined with a labelled transition system, the \emph{protocol specification language} of
types for message-passing objects, and the \emph{logical relation}.

\subsection{Computational Model}%
\label{sub:comp-model}

In order to provide a unified specification for any arbitrary object in our system, 
we define a \emph{labelled transition system} (LTS) \cite{MilnerBook1980, MilnerBook1999, SangiorgiWalkerBook2001}.
Defining our computational model in terms of an LTS allows us to \emph{abstractly} describe communications between different parties in a system,
making it particularly well-suited for a heterogeneous message-passing setting.
In an LTS, each transition is annotated with an \emph{action}, which describes the readiness of an object to engage in communication.
An action can be a readiness to \emph{send} a message, a readiness to \emph{receive} a message, and an empty action.
The empty action represents an actual message exchange, which occurs when two objects have \emph{complementary} actions (i.e. a send and a receive)
and are able to proceed with communication.

To provide some intuition for how the LTS can describe the communication between different parties in a system, consider the following example:

\begin{example}\label{ex:automaton}
  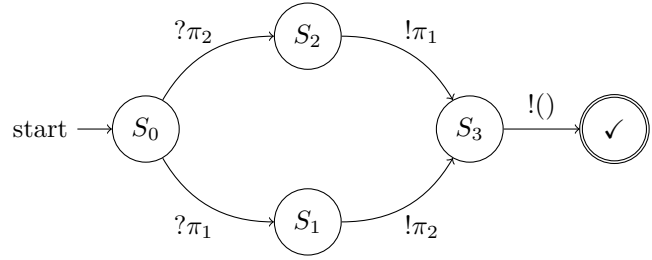
\begin{figure}[h]
    \centering
    \begin{tikzpicture}[]
        \node[state, initial] (S0) {$S_{0}$};
        \node[state, below right=of S0, xshift=5mm, yshift=4mm] (S1) {$S_{1}$};
        \node[state, above right=of S0, xshift=5mm, yshift=-4mm] (S2) {$S_{2}$};
        \node[state, below right=of S2, xshift=5mm, yshift=4mm] (S3) {$S_{3}$};
        \node[state, accepting, right=of S3] (S4) {$\checkmark$};

        \draw[->] (S0) edge[bend right, below] node[xshift=-2mm] {$?\pi_1$} (S1);
        \draw[->] (S0) edge[bend left,  above] node[xshift=-2mm] {$?\pi_2$} (S2);
        \draw[->] (S1) edge[bend right, below] node[xshift=2mm] {$!\pi_2$} (S3);
        \draw[->] (S2) edge[bend left,  above] node[xshift=2mm] {$!\pi_1$} (S3);
        \draw[->] (S3) edge node[above] {$! ()$} (S4);
    \end{tikzpicture}
    \caption{Bit-flipping automaton}
    \label{fig:bit-flipping-automaton}
\end{figure}

  The diagram in \cref{fig:bit-flipping-automaton} describes a bit-flipping automaton, which is akin to a signalling object in an Internet of Things (IoT) application. We start at state $S_0$, and we have two possible transitions from that state,
  labelled $? π_1$ and $? π_2$. These labels describe the transition for when the state $S_0$ \emph{receives} a boolean-valued label of either $π_1$ (\textcode{false}) or $π_2$ (\textcode{true}).
  In the former case, it will transition to state $S_1$, and in the latter to $S_2$. After this, the states $S_1$ and $S_2$ have transitions labeled with $! π_2$ and $! π_1$ (respectively), which describe
  the \emph{sending} of those boolean-valued labels in either case. Once the send is performed, both states transition to $S_3$, which then transitions to a terminal state $\checkmark$
  with the label $! ()$, sending a \emph{closing signal} to indicate the end of all communication.
\end{example}

With this intuition in mind, we proceed with formally defining the LTS, and we begin by stating all the relevant definitions for \emph{actions}.
Actions rely on the notion of a \emph{channel},
which is an identifier that a client (e.g. a controller) who wishes to communicate with an object (e.g. a sensor) chooses for the object.
Channels thus serve as the object's point of contact and can be part of a message's content, i.e. its \emph{payload}:

\begin{defn}[Payload]
We define a \emph{payload} to be one of the following:
\begin{itemize}
\item A boolean-valued \emph{selector}, denoted $π_1$ or $π_2$, or
\item A \emph{closing signal} denoted $()$, or
\item A \emph{channel name} $a ∈ \Labels$.
\end{itemize}
\end{defn}

Note that the set $\Labels$ is some countably infinite set of identifiers. We may now proceed to formally define \emph{actions}:

\begin{defn}[Action]
We define an \emph{action} to be one of the following:
\begin{itemize}
\item The constant $ε$, or
\item A triple $(a, d, p)$ where $a ∈ \Labels$ is a channel name,
$d$ is the \emph{direction}, which can be either $!$ (sending) or $?$ (receiving), and $p$ is a payload.
\end{itemize}
Note that the constant $ε$ stands for an empty action (also known as a \emph{silent transition}).
This corresponds to $τ$-transitions in $π$-calculus~\cite{MilnerBook1980,MilnerBook1999,SangiorgiWalkerBook2001}.
\end{defn}

In our computational model, the smallest unit of computation is an \emph{atomic process},
which can communicate with other processes through channels,
akin to processes in the $\pi$-calculus~\cite{MilnerBook1980,MilnerBook1999,SangiorgiWalkerBook2001}.
So far we have been referring to such a unit colloquially as an ``object.''

\begin{defn}[Atomic Process]
We define a set of \emph{atomic process}, written $\AtomicObj{A}$,
which has the following elements:
\begin{itemize}
\item $\procEx a {\nvProc}$, for a channel name $a ∈ \Labels$ and an element $\nvProc ∈ A$, or
\item $\procFwd a b$, for a pair of channel names $a, b ∈ \Labels$.
\end{itemize}
Atomic processes are parameterized by a set $A$,
which we define next.
\end{defn}

Each atomic process has a \emph{providing channel}, which is the process' identifier (as chosen by the client) and acts as a point of contact. The channel $a$ in both $\procEx a {\nvProc}$
and $\procFwd a b$ is the providing channel.
The set $A$ contains process terms parameterized over a providing channel. In this sense,
the elements of $A$ are considered to be ``nameless,'' and become ``named'' once a providing
channel is supplied through the $\procEx a {-}$ construct.

We now define the \emph{process language} structure, which we use to abstract over the components of a heterogeneous system. Each structure
defines a set of \emph{nameless objects} ($\NObj$), which describes the ``terms'' in the language that are able to interact with a supplied channel,
and a \emph{transition relation} ($\stepsTo{\cdot}_\Obj$) on those terms, describing how named terms transition to (possibly many) other atomic processes according to some action.
The choice to have \emph{nameless} terms at the language level not only simplifies the treatment of identifiers in our development,
but also allows objects to be defined independently of a particular choice of providing channel (i.e. they are \emph{polymorphic} in the providing channel).
The latter aspect is especially apt for our target domain. The names of the components in  
heterogeneous systems often take the form of \emph{addresses}, such as IP addresses. These addresses 
are usually assigned as the system is wired up, once the components themselves have already been programmed. 
Conceptually, the functionality of the components are---and should be---independent of the names of the objects. 
This dictum is reflected by our use of nameless objects, which is inspired by the notion of
\emph{nameless family of configurations} in Yue \textit{et. al.} \cite{YaoPOPL2025}.

The definition of the process language structure relies on the notion of a finite multiset, which is given first:

\begin{defn}[Finite Multisets (FMSet)]
We define a \emph{finite multiset} over a set $X$, denoted by $\FMSet{X}$,
to be a finite set with possibly repeated elements drawn from $X$.
We denote it as $\FMSet{X}$, along with a disjoint union operator $\fmconcat$ with unit $\varnothing$.
\end{defn}

\begin{defn}[Process Language]\label{defn:proc-lang}
We define a process language to be a structure with the following data:
\begin{itemize}
\item A set of \emph{nameless objects} $\NObj$, which can be used as the argument of $\AtomicObjPrefix$, and
\item A \emph{transition relation} $\stepsTo{\cdot}_\Obj$ with the following type:
\[
\stepsTo{\cdot}_\Obj : (\NObj × \Labels) × \Actions × \FMSet{\AtomicObj{\NObj}} → \textsf{Prop}
\]
\end{itemize}
We also define the following notation:
\begin{itemize}
\item \ProcLang{} for the set of process language structures in the system,
\item $\vProc \stepsTo{a!c}_\Obj \vCfg$ to say that an object $\vProc$ and a set of atomic objects $\vCfg$ satisfy the relation
  $\stepsTo{\cdot}_\Obj$ where the arguments are a label $a$, the direction $!$,
  and a payload $c$.
\item $\vProc \stepsTo{a?c}_\Obj \vCfg$ to say that an object $\vProc$ and a set of atomic objects $\vCfg$ satisfy the relation
  $\stepsTo{\cdot}_\Obj$ where the arguments are a label $a$, the direction $?$,
  and a payload $c$.
\item $\vProc \stepsTo{ε}_\Obj \vCfg$ to say that an object $\vProc$ and a set of atomic objects $\vCfg$ satisfy the relation
  $\stepsTo{\cdot}_\Obj$ where the argument is the empty action $ε$.
\end{itemize}
\end{defn}

To represent a \emph{heterogeneous} set of components, we use the $Σ$-type
$\sum_{S ∈ \ProcLang}S.\NObj$. Additionally, we write:
\begin{itemize}
\item $\NObj$ as a shorthand for the type $\sum_{S ∈ \ProcLang}S.\NObj$,
\item $\Obj$ as a shorthand for the type $\AtomicObj{\NObj}$,
\end{itemize}
This $Σ$-quantification of the language means the processes can possibly be from different process languages.

With all of these definitions in place, we can now define the operational model our framework relies on:

\begin{defn}[Runtime System]\label{defn:runtime-system}
We define a \emph{runtime system} to consist of the following data:
\begin{itemize}
\item The multiset of objects $\Cfg := \FMSet{\Obj}$, called \emph{configurations},
  as well as the nameless version of them, $\NCfg := \Cfg × \NObj$,
\item An instantiation operation $\nvCfg[a]$ for $\nvCfg ∈ \NCfg$ and $a ∈ \Labels$,
  which destructs $\nvCfg$ as $\vCfg × \nvProc$ and returns $\vCfg \fmconcat{} \procEx a {\nvProc}$,
\item A \textit{transition relation} $\stepsTo{\cdot}$ between $\Cfg$,
  given by the rules shown in~\cref{fig:cfg-rules},
\item A multistep version $\stepsTo{}^*$ of the transition relation where the action is $ε$, and
\item A \textit{nameless transition relation} $\stepsTo{}_\NCfg$ between elements of type $\NCfg$.
\end{itemize}
\end{defn}

\begin{figure}[h!]
\centering
\begin{mathpar}
\inferrule[Step-Obj]{\vProc \stepsTo{α}_\Obj \vCfg'}{\set{\vProc} \stepsTo{α} \vCfg'} \and
\inferrule[Step-Fwd]{ ~ }{ \Set{ \procEx a {\nvProc}, \procFwd b a } \stepsTo{α} \Set{ \procEx b {\nvProc} } } \and
\inferrule[Step-Frame]{\vCfg \stepsTo{α} \vCfg'}{\vCfg \fmconcat \vCfg_0 \stepsTo{α} \vCfg' \fmconcat \vCfg_0} \and
\inferrule[Step-Comm]{\vCfg_1 \stepsTo{a!c} \vCfg_1' \\ \vCfg_2 \stepsTo{a?c} \vCfg_2'}{\vCfg_1 \fmconcat \vCfg_2 \stepsTo{} \vCfg_1' \fmconcat \vCfg_2'}
\end{mathpar}
\caption{Configuration stepping rules}
\label{fig:cfg-rules}
\end{figure}

The transition rules in~\cref{fig:cfg-rules} are as follows:
\begin{enumerate}
\item \textsc{Step-Obj}: An inclusion from language-level transition to configuration-level transition,
\item \textsc{Step-Fwd}: The operational semantics of forwarding processes,
\item \textsc{Step-Frame}: The \emph{frame rule}, similar to the one in $π$-calculus \cite{MilnerBook1980, MilnerBook1999,SangiorgiWalkerBook2001}, and
\item \textsc{Step-Comm}: The \emph{communication rule}, which performs a message exchange.
\end{enumerate}

The nameless transition relation $\stepsTo{}_\NCfg$ is defined by:
\begin{itemize}
\item $\nvCfg \stepsTo{}_\NCfg \nvCfg'$ if and only if $∀ a.~\nvCfg[a] \stepsTo{} \nvCfg'[a]$, and
\item $\nvCfg \stepsTo{}_\NCfg^* \nvCfg'$ if and only if $∀ a.~\nvCfg[a] \stepsTo{}^* \nvCfg'[a]$
\end{itemize}

If a process steps to an empty multiset, we consider it to be a process that \textit{terminates},
and it disappears from the environment.

\begin{lem}[Multistep Frame]
The multistep version of the frame rule is admissible:
\begin{mathpar}
\inferrule{\vCfg \stepsTo{}^* \vCfg'}{\vCfg \fmconcat \vCfg_0 \stepsTo{}^* \vCfg' \fmconcat \vCfg_0} \and
\end{mathpar}
\end{lem}

\subsection{Protocol Specification Language}%
\label{sub:protocol-spec-language}

The protocols are specified by \emph{behavioral types} \cite{AnconaARITCLE2016,GayRavaraBOOK2017}, defined below:
\[
A, B ::= 1 \mid A ⊗ B \mid A ⊕ B \mid A \lwith B \mid A ⊸ B
\]

% from Yue's paper
The connectives include constructs for \emph{sequencing} messages and \emph{branching}. Sequencing is expressed with
the types $A ⊸ B$ and $A ⊗ B$. The type $A ⊸ B$ indicates that, after \emph{receiving} a channel of type $A$,
the protocol transitions to behaving as type $B$. Dually, type $A ⊗ B$ denotes the \emph{sending} of a
channel of type $A$ and proceeding as type $B$ afterwards. Branching is expressed by the types $A \lwith B$ and $A ⊕ B$. 
$A \lwith B$ \emph{offers} a choice between channels of type $A$ and $B$ (i.e. it is ready to \emph{receive} channels of either type),
and $A ⊕ B$ \emph{makes} a choice between channels of type $A$ and $B$ (i.e. it is ready to \emph{send} channels of either type).
The choice is conveyed by receiving and sending the labels $π_1$ or $π_2$, which are boolean-valued indicators.
The type $1$ denotes the terminal state of a protocol.

To provide some intuition for how these types specify message-passing protocols, consider the following example:

\begin{example}\label{ex:automaton-typing}
Recall the bit-flipping example from~\cref{ex:automaton} (as shown in \cref{fig:bit-flipping-automaton}).
We can specify the protocol of the bit-flipping automaton using \emph{behavioral types}. 
The type for the initial state $S_0$ of the bit-flipping automaton is $(1 ⊕ 1) \lwith (1 ⊕ 1)$, which indicates that the automaton can \emph{receive} either $\pi_1$ or $\pi_2$, and upon receiving will transition to a state behaving as type $1 ⊕ 1$, in either case. 
The type $1 ⊕ 1$ indicates that the automaton in that state can \emph{send} either $\pi_1$ or $\pi_2$, and afterwards transition to a state behaving as type $1$. 
Finally, the type $1$ indicates that the automaton is in a terminal state, at which point it can send the closing signal $()$ and end all communication.
\end{example}

\subsection{Logical Relation}%
\label{sub:logical-relation}

With both our operational model and specification language in place, we may now define our primary verification framework, which is the logical relation.
The \emph{logical relation} defines the intended runtime behavior of configurations for a given type, in a sense defining what it means for a configuration to be an ``inhabitant'' of that type and thus to comply with the protocol prescribed by the type.
For the following definitions, we assume that $\nvCfg$ is an element of the set $\NCfg$.
\begin{figure}[h!]
\centering
\begin{align}
\nvCfg ∈ \termInterp{A}
    \iff{}& ∃ \nvCfg' \text{ s.t. } \nvCfg \stepsTo{}_\NCfg^* \nvCfg' ∧ \nvCfg' ∈ \valueInterp{A} \label{termInterp} \\
\nvCfg ∈ \valueInterp{1}
    \iff{}& ∀a.~\nvCfg[a] \stepsTo{a!()} \varnothing \label{valueInterp:one} \\
\nvCfg ∈ \valueInterp{A ⊗ B}
    \iff{}& ∃a.~∀b.~∃ (\nvCfg_1 ∈ \termInterp{A}).~∃ (\nvCfg_2 ∈ \termInterp{B}).~\nvCfg[b] \stepsTo{b!a} \nvCfg_1[a] \fmconcat{} \nvCfg_2[b] \label{valueInterp:tensor} \\
\nvCfg ∈ \valueInterp{A ⊸ B}
    \iff{}& ∀a.~∀b.~∀ (\nvCfg_1 ∈ \termInterp{A}).~∃ (\nvCfg_2 ∈ \termInterp{B}).~\nvCfg[b] \fmconcat{} \nvCfg_1[a] \stepsTo{b?a} \nvCfg_2[b] \label{valueInterp:lolli} \\
\nvCfg ∈ \valueInterp{A \lwith B}
    \iff{}& (∀a.~∃ (\nvCfg_1 ∈ \termInterp{A}) . \nvCfg[a] \stepsTo{a ? π_1} \nvCfg_1[a]) ∧ \notag \\
        & (∀a.~∃ (\nvCfg_2 ∈ \termInterp{B}) . \nvCfg[a] \stepsTo{a ? π_2} \nvCfg_2[a]) \label{valueInterp:lwith} \\
\nvCfg ∈ \valueInterp{A ⊕ B}
    \iff{}& (∀a.~∃ (\nvCfg_1 ∈ \termInterp{A}) . \nvCfg[a] \stepsTo{a ! π_1} \nvCfg_1[a]) ∨ \notag \\
        & (∀a.~∃ (\nvCfg_2 ∈ \termInterp{B}) . \nvCfg[a] \stepsTo{a ! π_2} \nvCfg_2[a]) \label{valueInterp:lplus}
\end{align}
\caption{Definition of logical relation}
\label{fig:logical-relation}
\end{figure}

\cref{fig:logical-relation} gives the definition of the logical relation in terms of
two mutually recursive ``interpretations'' $\termInterp{A}$ and $\valueInterp{A}$,
often referred to as the ``term'' (or expression, computation) and ``value'' interpretations of type $A$.
Configurations $\nvCfg$ are in the term interpretation of $A$ if they perform some silent transitions (i.e. \emph{internal communications}), stepping to (in possibly multiple steps) some $\nvCfg' ∈ \valueInterp{A}$.
Configurations $\nvCfg$ are in the value interpretation of $A$ if they are ready to send or receive a message and thus engage in an \emph{external communication} with their client.
The value interpretation is defined by structural induction on the type,
specifying for each type what the expected runtime behavior of an element of that type is.

We now go through each case of the value interpretation, explaining how it captures the behavior specified in \cref{sub:protocol-spec-language} at each type.
Suppose we have some $\nvCfg ∈ \NCfg$. Then,
\begin{description}
  \item[\cref{valueInterp:one}] For all channels $a$, $\nvCfg$ is in the value interpretation at the type $1$ if
    $\nvCfg[a]$ sends a closing signal and transitions to the empty set.
  \item[\cref{valueInterp:tensor}] 
    $\nvCfg$ is in the value interpretation at the type $A ⊗ B$ if there exists a channel $a$ such that for all channels $b$,
    there exist some configurations $\nvCfg_1 ∈ \termInterp{A}$ and $\nvCfg_2 ∈ \termInterp{B}$, such that $\nvCfg[b]$ sends the channel $a$ to $b$ and transitions to the configuration containing $\nvCfg_1[a]$ and $\nvCfg_2[b]$.
    In other words, after sending the channel $a$ via the channel $b$, which will communicate something of type $A$, the configuration breaks down into a process that provides this communication (i.e. $\nvCfg_1[a]$) and a process that will continue afterwards (i.e. $\nvCfg_2[b]$).
  \item[\cref{valueInterp:lolli}]
    For all channels $a$ and $b$, $\nvCfg$ is in the value interpretation at the type $A ⊸ B$ if,
    for all configurations $\nvCfg_1 ∈ \termInterp{A}$, there exists some configuration $\nvCfg_2 ∈ \termInterp{B}$ such that $\nvCfg[b]$ with $\nvCfg_1[a]$ receives the channel $a$ at $b$ and transitions to $\nvCfg_2[b]$.
    In other words, channel $b$ receives the channel $a$, which provides a communication of type $A$ via the process $\nvCfg_1[a]$, and then continues with the process $\nvCfg_2[b]$ of type $B$.
  \item[\cref{valueInterp:lwith}] $\nvCfg$ is in the value interpretation at the type $A \lwith B$ if
    for all channels $a$, there exists some $\nvCfg_1 ∈ \termInterp{A}$ such that $\nvCfg[a]$ receives the boolean indicator $π_1$ and transitions $\nvCfg_1[a]$,
    \emph{and} for all channels $a$, there exists some $\nvCfg_2 ∈ \termInterp{B}$ such that $\nvCfg[a]$ receives the boolean indicator $π_2$ and transitions to $\nvCfg_2[a]$.
    In other words, the configuration $\nvCfg$ should be able to accommodate being sent a $π_1$ or $π_2$ across their providing channel, and will transition according to whichever it receives.
  \item[\cref{valueInterp:lplus}] $\nvCfg$ is in the value interpretation at the type $A ⊕ B$ if
    for all channels $a$, there exists some $\nvCfg_1 ∈ \termInterp{A}$ such that $\nvCfg[a]$ sends the boolean indicator $π_1$ and transitions $\nvCfg_1[a]$,
    \emph{or} for all channels $a$, there exists some $\nvCfg_2 ∈ \termInterp{B}$ such that $\nvCfg[a]$ sends the boolean indicator $π_2$ and transitions to $\nvCfg_2[a]$.
    In other words, the configuration $\nvCfg$ should either send a signal $π_1$ or $π_2$ across their providing channel, and will transition according to whichever it sends.
\end{description}

The relation is closed under pre-composition of stepping, as expected:
\begin{lem}[Backwards Closure]\label{lem:back-closure}
If $\nvCfg' ∈ \termInterp{A}$ and $\nvCfg \stepsTo{}^* \nvCfg'$, then $\nvCfg ∈ \termInterp{A}$.
\end{lem}

\section{Case Study}%
\label{sec:case-study}
In this section, we explore two verification scenarios that our framework is able to accomodate: 
(1) \emph{per-instance verification} of a signalling object in a heterogenous system, modeled as a bit-flipping automaton, and
(2) \emph{once-and-for-all verification} of arbitrary well-typed programs for a given type system.

\subsection{Per-instance verification}

Our verification framework allows us to show that any component of a system is ``well-behaved'' as prescribed by the logical relation.
To demonstrate this, we will show that our running example of the bit-flipping automaton (as introduced in \cref{ex:automaton} and illustrated in \cref{fig:bit-flipping-automaton})
inhabits the logical relation and formally prove that it fits the protocol specification defined in \cref{ex:automaton-typing}.

To start, we need to implement the automaton as a process language structure, as defined in \cref{defn:proc-lang}. 
\begin{defn}[Bit-Flipping Automaton Language Structure]\label{defn:bit-flip-automaton-lang-struct}
We specify all the necessary fields to populate a process language structure for the automaton:
  \begin{itemize}
    \item A set of nameless objects $\NObj$, specified as the set of all states except the accepting state in the automaton, i.e.
    \[ \NObj := \{S_0, S_1, S_2, S_3\} \]
    The accepting state is not included because the final transition (i.e. the transition that sends a closing signal) steps to the empty set.
    \item A transition relation $\stepsTo{\cdot}_\Obj$ consisting of the steppings listed in \cref{fig:obj-level-stepping-flip-bit}.
  \end{itemize}
\end{defn}

\begin{figure}
\begin{align*}
(S_0, a) & \stepsTo{a?π_1}_\Obj \set{\procEx{a}{S_1}} \\
(S_0, a) & \stepsTo{a?π_2}_\Obj \set{\procEx{a}{S_2}} \\
(S_1, a) & \stepsTo{a!π_2}_\Obj \set{\procEx{a}{S_3}} \\
(S_2, a) & \stepsTo{a!π_1}_\Obj \set{\procEx{a}{S_3}} \\
(S_3, a) & \stepsTo{a!()}_\Obj \varnothing
\end{align*}
\caption{Object-level stepping for the bit-flipping automaton}
\label{fig:obj-level-stepping-flip-bit}
\end{figure}

With this definition in place, all that remains is to show that our automaton inhabits the logical relation. 
To do this, we state and prove the following theorem:

\begin{theorem}\label{thm:inhabit} 
  Take $(\varnothing, S_0) ∈ \NCfg$. Then, 
  $ (\varnothing, S_0) ∈ \termInterp{(1 ⊕ 1) \lwith{} (1 ⊕ 1)}$.
\end{theorem}

In this theorem, $(\varnothing, S_0)$ represents the initial state of the automaton. This is an element of $\NCfg$, as defined in \cref{defn:runtime-system}.
The $\varnothing$ indicates that there are no other atomic objects around.
The theorem then concludes that the automaton inhabits the term interpretation of the logical relation at the type $(1 ⊕ 1) \lwith{} (1 ⊕ 1)$---that is, the automaton behaves ``correctly'' with respect to this specification.
We have already informally reasoned this to be the case in \cref{ex:automaton-typing}, but proving this theorem, we can \emph{certify} this to be the case.

\subsection{Once-and-for-all Verification}
\label{sec:once-and-for-all-verification}

% from Yue's paper
To perform this mode of verification, we need to develop a type system that is strong enough to ensure that any
well-typed term behaves as prescribed by the logical relation. The proof of this property is referred to as the \emph{fundamental theorem of the logical relation}
(FTLR). By carrying out the proof of the FTLR ``once-and-for-all,'' per-program verification reduces to a typechecking problem---if it typechecks, we can simply invoke the FTLR to get our desired result.
Additionally, if typechecking is decidable, proving the FTLR now makes program verification automatic.

For this case study, we will consider a language of process terms which can be specified with the types defined in \cref{sub:protocol-spec-language}.
The grammar for this language is defined in \cref{fig:term-grammar},
with the terms in the left column being constructors for the corresponding behavioral types in the right column.
The grammar contains boolean-valued signals $e$ (with $π_1$ being $\textcode{false}$ and $π_2$ being $\textcode{true}$), symbols $s$ (which can either be channel names $a ∈ \Labels$, where $\Labels$ is a countably infinite set of identifiers, or a variable $x$),
and process terms $M$.
Binding occurrences are denoted by $x ⇒ M$.
As detailed in the next section, we will type these process terms using
\emph{intuitionistic linear logic session types} \cite{CairesCONCUR2010,ToninhoESOP2013,ToninhoPhD2015}.

\begin{figure}[ht]
\begin{align*}
e ::={} & π₁ \mid π₂ \\
s ::={} & a \mid x \\
M ::={} & \fwdEx(s) \mid \letEx~x:A ← M₁; M₂ \\
  \mid{} & \sendEx() \mid \recvEx_x()
    && (\emph{for}~1) \\
  \mid{} & \recvEx(x ⇒ M) \mid \sendEx_s(s'); M
    && (\emph{for}~⊸) \\
  \mid{} & \sendEx(s); M \mid \recvEx_s(x ⇒ M)
    && (\emph{for}~⊗) \\
  \mid{} & \recvEx(π₁ ⇒ M₁ \mid π₂ ⇒ M₂) \mid \sendEx_s(e); M
    && (\emph{for}~\lwith) \\
  \mid{} & \sendEx(e); M \mid \recvEx_s(π₁ ⇒ M₁ \mid π₂ ⇒ M₂)
    && (\emph{for}~⊕)
\end{align*}
\caption{Grammar for the process language.}\label{fig:term-grammar}
\end{figure}

\subsubsection{Typing Rules}

\begin{figure}[h!]
\centering
\begin{mathparpagebreakable}
\inferrule[\textsc{Cut}]
  { Γ₁ ⊢ M₁ :: A \\
    Γ₂, x : A ⊢ M₂ :: B }
  { Γ₁, Γ₂ ⊢ (\letEx~x:A ← M₁ ; M₂) :: B }

\inferrule[\textsc{Id}]
  { }
  { x : A ⊢ \fwdEx(← x) :: A } \\

\inferrule[\textsc{1-Right}]
  { }
  { ⊢ \sendEx() :: 1 }

\inferrule[\textsc{1-Left}]
  { Γ ⊢ M :: C }
  { Γ, x:1 ⊢ (\recvEx_x(); M) :: C } \\

\inferrule[\textsc{$⊸$-Right}]
  { Γ, y : A ⊢ M :: B }
  { Γ ⊢ \recvEx(y ⇒ M) :: (A ⊸ B) }

\inferrule[\textsc{$⊸$-Left}]
  { Γ, x : B ⊢ M :: C } % the ``continuation''
  { Γ, x : A ⊸ B, y : A ⊢ (\sendEx_x(y) ; M) :: C }

\inferrule[\textsc{$⊗$-Right}]
  { Γ ⊢ M :: B }
  { Γ, y : A ⊢ (\sendEx(y); M) :: (A ⊗ B) }

\inferrule[\textsc{$⊗$-Left}] % uncurry
  { Γ, x : B, y : A ⊢ M :: C }
  { Γ, x : A ⊗ B ⊢ \recvEx_x(y ⇒ M) :: C }
  % x turns into B in the continuation

\inferrule[\textsc{$\lwith$-Right}]
  { Γ ⊢ M_1 :: A_1 \\
    Γ ⊢ M_2 :: A_2 }
  { Γ ⊢ \recvEx(π_1 ⇒ M_1 \mid π_2 ⇒ M_2) :: (A_1 \lwith{} A_2) }

\inferrule[\textsc{$\lwith$-Left}]
  { Γ, x : A_i ⊢ M :: C \\ i ∈ \set{1, 2} }
  { Γ, x : A₁ \lwith{} A₂ ⊢ (\sendEx_x(π_i); M) :: C }

\inferrule[\textsc{$⊕$-Right}]
  { Γ ⊢ M :: A_i \\ i ∈ \set{1, 2} }
  { Γ ⊢ (\sendEx(π_i); M) :: (A₁ ⊕ A₂) }

\inferrule[\textsc{$⊕$-Left}]
  { Γ, x : A ⊢ M₁ :: C \\
    Γ, x : B ⊢ M₂ :: C }
  { Γ, x : A ⊕ B ⊢ \recvEx_x(π₁ ⇒ M₁ \mid π₂ ⇒ M₂) :: C }

\end{mathparpagebreakable}

\caption{Typing rules for the process language.}%
\label{fig:typing-rules}
\end{figure}

To type the process terms defined in \cref{fig:term-grammar} we use \emph{session types} \cite{HondaCONCUR1993,HondaESOP1998,HondaPOPL2008}, and in particular \emph{intuitionistic linear logic session types}
\cite{CairesCONCUR2010,ToninhoESOP2013,ToninhoPhD2015}.
Session types are an instance of behavioral types \cite{AnconaARITCLE2016,GayRavaraBOOK2017}
with strong theoretical foundations, 
relating the session-typed $\pi$-calculus with linear logic
\cite{CairesCONCUR2010,WadlerICFP2012,ToninhoESOP2013,ToninhoPhD2015,LindleyMorrisESOP2015,KokkePOPL2019}.

To type the process terms in \Cref{fig:term-grammar}, we use a typing judgment of the form
$$Γ ⊢ M :: A,$$
which can be read as ``the process term $M$ provides the communication behavior (i.e. ``session'') of type $A$ under typing context $Γ$,'' where $Γ$ provides the typing of free variables in $M$.
The inference rules for this judgment are defined in~\cref{fig:typing-rules}
and are standard.

The rules in \cref{fig:typing-rules} are given in a \emph{sequent calculus},
describing the behavior of a session from the point of view of a provider of the session---expressed by so-called \emph{right rules}---as well as from the point of view of a client---expressed by so-called \emph{left rules}.
Computationally, the rules can be read bottom-up,
where the type of the conclusion denotes the protocol state of the object before the message exchange,
and the type of the premise denotes the protocol state after the message exchange.
The rules are thus in agreement with the behavior specified in \cref{sub:protocol-spec-language} at each type.
For example, in case of the right rule \textsc{$⊕$-Right}, the provider sends either of the labels $π_1$ or $π_2$
and then transitions to offering the session $A_i$, for $i ∈ \set{1, 2}$.
On the other hand, in case of the left rule \textsc{$⊕$-Left}, the client branches on the received label,
continuing with either $M₁$ or $M₂$.

While the typing judgment uses a context $Γ$ in a similar manner to the $λ$-calculus, the typing context here contains a list of \emph{bindings}, which we represent as a finite map
$Γ : \textsf{Vars} \partialmap \textsf{Types}$. In this definition, $\textsf{Vars}$ is the set of variables and $\textsf{Types}$ is the set of types.
Additionally, we use the conventional $λ$-calculus notation $Γ, x: A$ for the extension of contexts.
However, this should be interpreted as inserting the key-value pair $x \mapsto A$ into the map $Γ$.

\subsubsection{Runtime Process Terms}

At runtime, we assign channel names (i.e.~elements of the set $\Labels$) to variables in the context.
Such an assignment is called a \emph{substitution} $σ : \textsf{Vars} \partialmap \Labels$ from variables to names.
These substitutions work slightly differently from substitutions in the $λ$-calculus,
in that we are only substituting channel names for variables, not other terms with possibly free variables.
As a result, the possibility of incurring capture is ruled out in the first place.

Similar to the $λ$-calculus, we write $σ,b/x$ for extending the substitution $σ$ with the key-value
pair $x \mapsto b$, where $b$ is a channel name.
Note that this does not mean that substitutions are ordered in any way.
Additionally, we write $σ(x)$ for obtaining the name assigned to the variable $x$ from $σ$.
Thus, runtime process terms are obtained by applying a substitution to process terms $M$, which defined by structural induction on $M$.
The full definition can be found in~\cref{fig:subst}.
Notably, the operation $σ-x$ is used for removing bound variables $x$ from the substitution $σ$.

\begin{figure}[ht]
\begin{align*}
  \hat{σ}(\letEx~x:A ← M₁; M₂) &= \letEx~x:A ← \hat{σ}(M₁); \widehat{σ-x}(M₂) \\
  \hat{σ}(\fwdEx(← x)) &= \fwdEx(← σ(x)) \\
  \hat{σ}(\sendEx()) &= \sendEx() \\
  \hat{σ}(\recvEx_x(); M) &= \recvEx_{σ(x)}(); \widehat{σ-x}(M) \\
  \hat{σ}(\recvEx(y ⇒ M)) &= \recvEx(y ⇒ \widehat{σ-y}(M)) \\
  \hat{σ}(\sendEx_x(y); M) &= \sendEx_{σ(x)}(σ(y)); \hat{σ}(M) \\
  \hat{σ}(\sendEx(y); M) &= \sendEx(σ(y)); \hat{σ}(M) \\
  \hat{σ}(\recvEx_x(y ⇒ M)) &= \recvEx_{σ(x)}(y ⇒ \widehat{σ-x}(M)) \\
  \hat{σ}(\recvEx(π_1 ⇒ M₁ \mid π_2 ⇒ M₂)) &= \recvEx(π_1 ⇒ \hat{σ}(M₁) \mid π_2 ⇒ \hat{σ}(M₂)) \\
  \hat{σ}(\sendEx_x(π_i); M) &= \sendEx_{σ(x)}(π_i); \hat{σ}(M) \evidence{(for $i∈\set{1,2}$)} \\
  \hat{σ}(\sendEx(π_i); M) &= \sendEx(π_i); \hat{σ}(M) \evidence{(for $i∈\set{1,2}$)} \\
  \hat{σ}(\recvEx_x(π_1 ⇒ M₁ \mid π_2 ⇒ M₂)) &= \recvEx_{σ(x)}(π_1 ⇒ \hat{σ}(M₁) \mid π_2 ⇒ \hat{σ}(M₂))
\end{align*}
\caption{Substitution of terms in the process language.}\label{fig:subst}
\end{figure}

We then have the following lemmas on substitution:
\begin{lem}[Composition of Substitution]\label{lem:subst-comp}
For any term $M$, substitution $σ$, variable $x$, and channel name $b ∈ \Labels$:
\[
  \widehat{b/x}(\hat{σ}(M)) = \widehat{σ, b/x}(M)
\]
\end{lem}
\begin{proof}
By induction on the structure of $M$ and expanding the definition of $\hat{σ}$.
\end{proof}

In order to state~\cref{lem:discard} and later on~\cref{defn:compl},
we need to introduce the following definition,
which extends binary relations to elementwise relations on finite maps.
The goal is to related contexts and substitutions---both are finite maps from variables.

\begin{defn}\label{defn:map_forall2}
We say two finite maps $m_1, m_2 : X \partialmap Y$ are \emph{related}
by a relation $R : Y × Y → \textcode{Type}$ when for every $x ∈ X$,
one of the following holds:
\begin{itemize}
\item $m_1$ and $m_2$ are both undefined at $x$,
\item $m_1$ and $m_2$ are both defined at $x$ and $R(m_1[x], m_2[x])$ holds.
\end{itemize}
\end{defn}

\begin{defn}\label{defn:map_ops}
We introduce the following notations:
\begin{itemize}
\item For $m : X \partialmap Y$ and $f : Y → Y'$,
we write $f(m)$ for applying $f$ to the values of $m$, without changing the keys.
\item For $m : X \partialmap Y$ and $x : X$,
we write $m - x$ for removing the key $x$ from $m$.
\end{itemize}
\end{defn}

Here are some immediate consequences of~\cref{defn:map_forall2}:

\begin{lem}\label{lem:map_forall2-simple-facts}
If $m_1, m_2$ are related by $R$, then the following are true:
\begin{enumerate}
\item $m_1$ and $m_2$ have the same domain,
\item If $R \implies R'$, then $m_1$ and $m_2$ are related by $R'$,
\item Only empty maps are related to empty maps,
\item If $R(y, y') \implies R'(f(y), y')$, then $f(m_1)$ is related to $m_2$ by $R'$,
\item If $R(y, y') \implies R'(y, f(y'))$, then $m_1$ is related to $f(m_2)$ by $R'$.
\end{enumerate}
\end{lem}

A convenient way to represent the facts that two maps have the same domain is to
say that they are related by any suitable relation, so the above facts can be used to derive
other useful propositions about these maps. For example, we can define the following lemma for related type context maps $Γ$ and substitution maps $σ$:

\begin{lem}[Discard]\label{lem:discard}
For $Γ ⊢ M :: A$ and substitutions $σ, σ'$ such that $σ$ and $Γ$ are related maps,
and $σ$ disjoint with $σ'$,
\[ \widehat{σ \cup σ'}(M) = \hat{σ}(M) \]
\end{lem}
\begin{proof}
By induction on the typing derivation of $M$.
\end{proof}

\subsubsection{Stepping Rules}

The stepping rules of the language are defined to instantiate a process language structure, as defined in \cref{defn:proc-lang}.

\begin{defn}
We implement a process language (\cref{defn:proc-lang}),
with nameless objects $\NObj$ being the terms in~\cref{fig:term-grammar},
and stepping rules specified in~\cref{fig:dynamics}.
\end{defn}

In~\cref{fig:dynamics}, the channel name in the left-hand
side process is always universally quantified, for $a, b, c ∈ \Labels$.

\begin{figure}[h!]
\[\begin{array}{lll}
  ({\fwdEx(← b)}, a) & \stepsTo{ε}_\Obj & \procFwd a b {} \\
  ({\letEx~x:A ← M_1; M_2}, a) & \stepsTo{ε}_\Obj & \procEx{b'}{M_2} \fmconcat{} \procEx a {M_1}, \text{for}~b'∈ \Labels \\
  ({\sendEx()}, a) & \stepsTo{a!()}_\Obj & \varnothing \\
  ({\recvEx_a();M}, b) & \stepsTo{a?()}_\Obj & \procEx b M \\
  ({\recvEx(x ⇒ M)}, b) & \stepsTo{b?a}_\Obj & \procEx b {\widehat{a/x}(M)} \\
  ({\sendEx_b(a);M}, c) & \stepsTo{b!a}_\Obj & \procEx c M \\
  ({\sendEx(a);M}, b) & \stepsTo{b!a}_\Obj & \procEx b M \\
  ({\recvEx_b(x ⇒ M)}, c) & \stepsTo{b?a}_\Obj & \procEx c {\widehat{a/x}(M)} \\
  ({\recvEx(π_1 ⇒ M_1 \mid π_2 ⇒ M_2)}, a) & \stepsTo{a?π_i}_\Obj & \procEx a {M_i}, \text{for}~i∈\set{1,2} \\
  ({\sendEx_a(π_i);M}, b) & \stepsTo{a!π_i}_\Obj & \procEx b M, \text{for}~i∈\set{1,2} \\
  ({\sendEx(π_i);M}, a) & \stepsTo{a!π_i}_\Obj & \procEx a M, \text{for}~i∈\set{1,2} \\
  ({\recvEx_a(π_1 ⇒ M_1 \mid π_2 ⇒ M_2)}, b) & \stepsTo{a?π_i}_\Obj & \procEx b {M_i}, \text{for}~i∈\set{1,2}
\end{array}\]
\caption{Object-level stepping of the process language.}%
\label{fig:dynamics}
\end{figure}

\subsubsection{Fundamental Theorem}

Before we can state the FTLR, we need some auxiliary definitions related to the typing context $Γ$. In particular,
we define the notion of \emph{complementary configurations} $S$ for a context $Γ$, which is a finite map from 
the variables in $Γ$ to values $(\nvCfg, a) ∈ \NCfg × \Labels$. We can consider each pair $(\nvCfg, a)$ to be
an object with its providing channel, and the map $S$ connects variables from $Γ$ to these objects as identified by their providing channel. 
At runtime, these variables will be substituted with channel names,
so we need a way to guarantee that these connected objects are also ``well-behaved'' as dictated by the logical relation.
This gives rise to the following definition:

\begin{defn}[Complementary Configurations]\label{defn:compl}
We define the set of \emph{complementary configurations}
for a context $Γ$ as the finite map $S$ valued in $\NCfg × \Labels$
such that $S$ and $Γ$ are related (\cref{defn:map_forall2}) by the logical relation $\termInterp{-}$.

In that case, we say $S$ complements $Γ$, denoted $S ∈ \termInterp{Γ}$.
\end{defn}

We may also apply the second projection to $S$ to obtain a substitution.
Using the notation in~\cref{defn:map_ops}, we may write $S.2$ for the substitution.

We also have the following lemma for extending complementary configurations with variables not contained in $Γ$:

\begin{lem}\label{lem:compl}
Suppose $\nvCfg ∈ \termInterp{A}$, $a ∈ \Labels$, $S ∈ \termInterp{Γ}$, and $x ∉ Γ$.
Then, \[(S, \set{x\mapsto (\nvCfg, a)}) ∈ \termInterp{Γ, x:A}.\]
\end{lem}
\begin{proof}
This follows straightforwardly from~\cref{lem:map_forall2-simple-facts}.
\end{proof}

The following definition describes the operation of ``applying'' complementary configurations $S$
to a nameless atomic process $\nvProc$. This operation, in a sense,
``links'' a nameless atomic process with all of the configurations it communicates with
by instantiating each of the $\nvCfg$s in the $S$ map with their providing channel.

\begin{defn}[applyCompl]\label{defn:applyCompl}
We may \emph{apply} complementary configurations $S ∈ \termInterp{Γ}$
to a nameless atomic process $\nvProc$ by taking the codomain of $S$
as a finite multiset of $\NCfg × \Labels$, and for each element $(\nvCfg, a)$ in the multiset,
we turn it into $\nvCfg[a]$ and union them all into a configuration $\vCfg$.
We then put this result together with $\nvProc$ to obtain $(\vCfg, \nvProc) ∈ \NCfg$.

We denote this operation as $\applyCompl{S}{\nvProc}$.
\end{defn}

With all of these definitions in place, we may now define the FTLR, which will allow us to perform the ``once-and-for-all'' verification for our process language.

\begin{theorem}[FTLR]\label{thm:ftlr}
For all typing context $Γ$, process term $M$, type $A$, and finite map $S$,
if $Γ ⊢ M :: A$ and $S ∈ \termInterp{Γ}$, then
$\applyCompl{S}{\hat{σ}(M)} ∈ \termInterp{A}$ where $σ = S.2$.
\end{theorem}

We then have the following corollary, which states that if a nameless atomic process $\nvProc$ provides a communication of type $1$ in an empty context,
then there must exist a configuration it steps to which can send a closing signal and terminate.
This result amounts to our \emph{adequacy} result, which is a way for us to ``check'' that inhabitants of the logical relation indeed have the expected property.

\begin{cor}[Adequacy]\label{cor:adequacy}
For channel name $a$, if $[] ⊢ \nvProc :: 1$ then there exists $Ω$ such that
$\procEx a {\nvProc} \stepsTo{}^* Ω$ and $Ω \stepsTo{a!()} []$.
\end{cor}
\begin{proof}
By FTLR, $\nvProc ∈ \termInterp{1}$, and by definition of term interpretation.
\end{proof}

\section{Mechanization}%
\label{sec:mechanization}
The above results, notably~\cref{thm:ftlr,cor:adequacy,thm:inhabit}, are mechanized in Coq.
We heavily utilize the \textcode{gmap} and \textcode{gmultiset} data structures from the \textlibraryname{stdpp} library~\cite{gmap, stdpp}.
The context $Γ$ in the typing judgment $Γ ⊢ M :: A$ is represented as a \textcode{gmap}
from variable names to types, and configurations are represented as \textcode{gmultiset}s of atomic processes.

% LangName
Recall the $Σ$-type used in~\cref{sec:verification-framework}:
\[ \sum_{S : \ProcLang}S.\NObj, \]
which we put into a multiset.
Unfortunately, \textcode{gmultiset} requires its elements to be countable and have decidable equality,
and for a $Σ$-type, this means that (at least) both projections must also have those properties.
It is very difficult to prove countability for an instance of $\ProcLang$,
because the dynamics of the language is an indexed proposition, which cannot be meaningfully encoded as a positive number.
So, we came up with an alternate representation to work with the requirements of \textcode{gmultiset}.
Instead of using
\( \sum_{S : \ProcLang}S.\NObj, \)
we define the datatype \textcode{LangName} to enumerate the ``tags'' of all the process language structures in $\ProcLang$,
and a function $\textcode{LangTable} : \textcode{LangName} → \ProcLang$
to their respective structures. We then redefine the $Σ$-type as the following:
\[ \sum_{l: \textcode{LangName}} \textcode{LangTable}(l).\NObj. \]
To guarantee extensibility,
we do not case on the \textcode{LangName} anywhere besides in the \textcode{LangTable} function
and the proof that \textcode{LangName} itself is countable.
Therefore, to extend the mechanization with new languages,
all that one needs to do is define additional tags in the \textcode{LangName} enumeration,
specify their process language structure definitions in the \textcode{LangTable} function,
and reprove its countability.

\section{Discussion and Future Work}%
\label{sec:discussion-future}
We now briefly reflect upon our results
and discuss possible avenues for future work.

% Quotients
% The mechanization uses maps and multisets from the \textlibraryname{stdpp} library,
% which relies on countability of the map keys and multiset elements.
% In a type theory with native support for quotient types~\cite{PujetOTT,ABCFHL,CHM},
% there can be a simpler implementation for multiset and maps,
% which would not rely on countability and decidable equality.
% Our proof does not essentially rely on these two properties,
% so using another proof assistant has the potential to simplify the mechanization
% and ease the burden of implementing a new language on top of the framework.

% Strong normalization
% \sbtodo{Question from Tesla: What should I say about strong normalization, freshness of channel names, etc.?
%   I can't think of anything concretely.}

% Variables
% In the mechanization, we used named variables. An obvious consequence is that
% every time we perform a context extension $Γ, x : A$, we need the additional assumption that $x \notin Γ$.
% Thankfully, we did not need judgmental equality, otherwise we will also need to deal with $α$-equivalence.
% Fancier solutions such as \emph{leftover typing}~\cite{GallaisLeftovers} or marking variable usages
% in the judgments could be a next step to simplify the mechanization.
% It can also prepare it for more complex logical relations in the future, such as binary ones.

% Substitution
The handling of binding and scope of variables,
elegantly taken care of by reliance on $\alpha$-equivalence classes in pen-and-paper proofs,
must be implemented in all its glory when it comes to mechanizations using proof assistants.
Of particular concern is the handling of substitution.
Thankfully, we do not run into the issue of capture of free variables when substituting
because we only ever substitute channels for variables.
However, when implementing the lifting of the substitution $σ$ over process terms $M$ defined in \cref{fig:subst},
% sonya: this part of the sentence didn't really flow, and really the only relevant part is that we're implementing substitution
% and used in the fundamental theorem \cref{thm:ftlr},
we had to make several implementation choices in this respect.
As is usual, we define $\hat{σ}$ by structural induction on the process term $M$.
Notably, these terms are \emph{raw}, and as such untyped.
The definition of $\hat{σ}$ is such that any bound variable encountered
is removed from the substitution $σ$,
using the operation $σ-x$.
Again, this choice is pretty standard; alternatively we could have updated the substitution $σ$
by mapping the bound variable to itself.

If we were to consider typed terms instead of raw terms,
the question comes up how $\hat{σ}$ should handle free variables
for our choice of type system in \cref{sec:once-and-for-all-verification}.
For instance, consider the typing rule for \textsc{$⊗$-Right} (with the explicit assumption $y \notin Γ$ present in our mechanization):
\begin{mathpar}
\inferrule[\textsc{$⊗$-Right}]
  { Γ ⊢ M :: B \\ y \notin Γ }
  { Γ, y : A ⊢ (\sendEx(y); M) :: (A ⊗ B) }
\end{mathpar}
By linear typing, we know that the variable $y$ must not be free in $M$. However, \cref{fig:subst} defines $\hat{σ}$ for this case as $\hat{σ}(\sendEx(y); M) = \sendEx(σ(y)); \hat{σ}(M)$.
If we were to define $\hat{σ}$ for well-typed terms, then it seems more appropriate to define $\hat{σ}$ as $\sendEx(σ(y)); \widehat{σ-y}(M)$.
As of now, we are unclear about which approach is favorable,
especially for the \textsc{Cut} rule.
% Defining $\hat{σ}$ for typed terms will necessitate the consideration of both variables \emph{and} channel names,
% which comes with its own complications.
We leave a careful exploration of this alternative approach for future work.

% but the substitution operation takes raw terms, so it is not very good to rely on typing in its definition.
% This causes a mismatch between the proof goal in FTLR and the induction hypothesis,
% \sbtodo{Question from Tesla: How detailed should I go? I can sketch the mismatch, is it a good idea?}
% and needs to be fixed by~\cref{lem:discard}.
% A better approach is to make substitution preserve well-typedness,
% but it is also difficult to define the well-typedness of the \emph{result} of substitution:
% these terms will involve in both variables \emph{and} channel names due to the existence of bound variables.
% This is a challenging problem that we leave for future work.

\section{Related Work}%
\label{sec:related}
We discuss related work in order of relatedness.
% sonya: changed capitalization of paragraphs to be consistent with intro
\subparagraph*{Mechanizations of logical relations for session types.}

Gollamudi \textit{et al.} \cite{GollamudiCoqPL2025} contributed the first mechanization of a semantic
logical relation for session types. Our work differs from theirs in the following key aspects: (1) 
\textit{Op. cit.} revolves around a singular and distinguished intuitionistic linear logic session type
language. This is evident from the explicit use of a process term syntax in the definition of their 
logical relation. Additionally, processes in their setting synchronize by juxtaposing complementary
process terms. In contrast, processes in our work synchronize through complementary actions, which 
is critical for language interoperability. (2) \textit{Op. cit.} enforces a strict rooted-tree 
communication structure with their support of scoped channels. While this provides an elegant method of 
handling name allocation, this restriction is detrimental for our purposes. Our work recognizes the value
of decoupling name assignment for a process from its semantic characterization, and in response, we propose
nameless objects. Our use of nameless objects reaps the majority of the benefits afforded by 
\textit{op. cit.}'s scoped channels, while preserving the flexibility of communication structures.
Our nameless objects are inspired by the notion of
nameless family of configurations in \cite{YaoPOPL2025},
with our work being the first to mechanize this notion.

\subparagraph*{Multi-language verification with processes.}  

Our work echos the development in DimSum \cite{SammlerPOPL2023} in a number of ways.
\textit{Op. cit.}, similarly to our work, proposes a process algebra framework for 
multi-language semantics and verification. In this framework, different languages are 
separated into different \emph{modules}, and the modules communicate with each other 
through synchronization of \emph{events}. Our work also models different languages as 
separate processes and allows them to synchronize over a common process substrate with 
messages. However, compared to \textit{op. cit.}, our work can be viewed as 
an alternative take on the idea of multi-language verification with processes where we 
trade flexibility for simplicity, and most importantly uniform composability.

DimSum is a very flexible framework: it allows the user to define different sets of events for 
different languages. However, this flexibility comes at a cost. During verification, 
the user must carefully set up the embedding of events across languages. Given a 
collection of modules, the user must carry out a linkage proof by hand using provided 
proof rules. 

In our system, the communication primitives available are fixed across languages. From 
the point of DimSum, we have chosen a fixed, but rather rich, set of events shared by 
all languages. This allows us to (1) establish fundamental theorems (e.g. \Cref{thm:ftlr}) that automate 
composition across languages and (2) support higher-order channels, where the name of 
a message-passing entity is passed as data to another entity.

\subparagraph*{Logical relations for session types.}

While our work contributes the second mechanization of logical relations for session types,
the theory of logical relations for session types is an active area of research.
The focus so far has been predominantly on unary logical relations
for proving termination \cite{PerezESOP2012, PerezARTICLE2014,DeYoungFSCD2020, RochaCairesESOP2023}.
Our work shares with this line of work the focus on proving protocol compliance and thus termination,
but differs in that all of our results are mechanized.
More recently,
binary logical relations for session types have been developed
for parametricity \cite{CairesESOP2013} and
noninterference \cite{DerakhshanLICS2021,DerakhshanECOOP2024,VanDenHeuvelECOOP2024} as well as
for program equivalence in general \cite{BalzerARXIV2023}
with support of general recursive types.
Of particular interest are the developments by Derakhshan \textit{et al.} \cite{DerakhshanLICS2021}
and Balzer \textit{et al.} \cite{BalzerARXIV2023},
which index the binary logical relation with an intuitionistic sequent,
rather than a single type,
and thus are able to semantically capture the resource semantics arising from linear logic.
We would like to explore mechanizing such an approach as part of future work,
which is still outstanding.

\section{Conclusion}
\label{sec:conlusion}
In this paper, we have defined, utilized, and mechanized a verification framework for heterogeneous applications.
Our language-agnostic approach allows us to employ this framework in the certification of the heterogeneous message-passing systems of today,
ranging from multi-language systems to systems consisting of both hardware and software components. 
The mechanization of this framework not only verifies the methods and proofs presented in this paper,
but provides the foundation for scaling this approach to further applications as the computing landscape continues to grow.

\bibliography{ref}

\appendix

\end{document}